\newif\ifpublic
\publictrue 

\ifpublic
\documentclass[10pt,a4paper]{article}
\else
\documentclass[10pt,draft,a4paper]{article}
\fi

\usepackage[pagebackref,colorlinks=true,linkcolor=blue,urlcolor=blue,citecolor=blue,pdfstartview=FitH]{hyperref}
\usepackage{fullpage}
\usepackage{amsmath}
\usepackage{amssymb}
\usepackage{amsthm}
\usepackage[russian,english]{babel}

\usepackage[utf8]{inputenc}
\usepackage{blkarray}
\usepackage{bm}
\usepackage{mathtools}
\usepackage{bbm} 
\usepackage{nicefrac,xfrac} 

\ifpublic
\usepackage[disable]{todonotes}
\else
\usepackage[colorinlistoftodos]{todonotes}
\fi

\newcommand*\samethanks[1][\value{footnote}]{\footnotemark[#1]}

\usepackage[capitalize,nameinlink]{cleveref}
\renewcommand{\eqref}[1]{\hyperref[#1]{(\ref*{#1})}}

\theoremstyle{plain}
\newtheorem{theorem}{Theorem}[section]

\newtheorem{lemma}[theorem]{Lemma}

\newtheorem{proposition}[theorem]{Proposition}
\newtheorem{definition}[theorem]{Definition}
\newtheorem{claim}[theorem]{Claim}

\theoremstyle{definition}

\renewcommand{\epsilon}{\varepsilon}
\newcommand{\eps}{\epsilon}
\newcommand{\zo}{\{0,1\}}
\newcommand{\OR}{\mathrm{OR}}
\newcommand{\AC}{\mbox{\rm AC}}

\newcommand{\reals}{\mathbb{R}}

\newcommand{\disjunion}{\mathbin{\dot{\cup}}}
\newcommand{\union}{\mathbin{\cup}}

\newcommand{\pdeg}{\operatorname{P-deg}}
\newcommand{\hcpdeg}{\operatorname{hcP-deg}}
\newcommand{\supp}{\operatorname{supp}}
\newcommand{\distP}{\mathbf{P}}
\newcommand{\distL}{\mathbf{L}}
\newcommand{\calD}{{\mathcal{D}}}

\newcommand{\calL}{{\mathcal{L}}}

\newcommand{\prob}[2]{\Pr_{#1}\left[ #2 \right]}
\newcommand{\bra}[1]{\{#1\}}
\newcommand{\cbra}[1]{\left(#1\right)}
\newcommand{\bigo}[1]{O\left(#1\right)}
\newcommand{\half}{\nicefrac{1}{2}}
\newcommand{\E}[2]{\mathop{\mathbb{E}}_{#1}\left[ #2 \right]}

\renewcommand{\calL}{\mathcal{L}}

\newcommand{\calR}{\mathcal{R}}
\newcommand{\lbdetail}{\frac{\log\cbra{\nicefrac1{8\eps}}\cdot\cbra{\log
      n - \llepsinv}}{2C \log^2\cbra{\nicefrac1{C^2}\cdot
      \log^4\cbra{\nicefrac1{8\eps}}\cdot \cbra{\log n -
        \llepsinv}^3}}}
\newcommand{\llepsinv}{\log\log(\nicefrac1\eps)}
\newcommand{\lepsinv}{\log(\nicefrac1\eps)}

\newcommand{\Omegatilde}[1]{\widetilde\Omega \cbra{#1}}
\newcommand{\bigomega}[1]{\Omega\cbra{#1}}
\newcommand{\quarter}{\nicefrac{1}{4}}
\newcommand{\zs}{(0,*)}
\newcommand{\ans}{\log \binom{n}{\leq \lepsinv}}

\newcommand{\lb}{\Omega \cbra{\frac{\ans}{\log^2 \cbra{\ans}}}}
\newcommand{\lbinline}{\Omega \cbra{\ans / \log^2 \cbra{\ans}}}

\title{On the Probabilistic Degree of OR over the Reals\thanks{A
    preliminary version of this paper appeared in {\em Proc. $38$th IARCS Annual Conf.\ on Foundations of Software Tech.\ and
  Theoretical Comp.\ Science (FSTTCS)} 2018~\cite{BhandariHMS2018}.}}

\author{
  Siddharth Bhandari\thanks{Tata Institute of Fundamental Research,
    INDIA. email: {\tt
      \{siddharth.bhandari,prahladh,tulasi.molli\}@tifr.res.in}. Research
  of the first and second author supported in part by the Google PhD Fellowship and Swarnajayanti Fellowship respectively.}
  \and Prahladh Harsha\samethanks 
  \and Tulasimohan Molli\samethanks 
  \and Srikanth Srinivasan\thanks{Department of Mathematics, IIT Bombay, INDIA. email: {\tt srikanth@math.iitb.ac.in}. Supported in part by SERB Matrics grant MTR/2017/000958.}}

\date{}

\begin{document}
\maketitle
\begin{abstract}
We study the probabilistic degree over $\reals$ of the $\OR$ function
on $n$ variables. For $\eps \in (0,1/3)$, the $\eps$-error probabilistic degree of 
any Boolean function $f:\zo^n\to \zo$ over $\reals$ is the smallest non-negative
integer $d$ such that the following holds: there exists
a distribution $\distP$ of polynomials $P(x_1,\ldots,x_n)  \in \reals[x_1,\ldots,x_n]$ of degree at most $d$ such that for all $\bar{x}
\in \zo^n$, we have $\Pr_{P \sim \distP}[P(\bar{x}) = f(\bar{x}) ] \geq 1- \eps$.  It is
known from the works of Tarui ({\em Theoret. Comput. Sci.} 1993) and
Beigel, Reingold, and Spielman ({\em Proc.\ $6$th CCC} 1991), that the
$\eps$-error probabilistic
degree of the $\OR$ function is at most $O(\log n \cdot \lepsinv)$.  Our first observation
is that this can be improved to $\bigo{\ans}$ 
which is better for small values of $\eps$. 

In all known constructions of probabilistic polynomials for the $\OR$
function (including the above improvement), the polynomials $P$ in the
support of the distribution $\distP$ have the
following special structure:
\[ 
P(x_1,\dots,x_n) = 1 - \prod_{i \in [t]} \left(1- L_i(x_1,\dots,x_n)\right),
\]
where each $L_i(x_1,\dots, x_n)$ is a linear form in the variables
$x_1,\ldots,x_n$, i.e., the polynomial $1-P(\bar{x})$ is a product of affine
forms. We show that the $\eps$-error probabilistic degree of
$\OR$ when restricted to polynomials of the above form 
is $\Omega\left( \ans/\log^2 \left( \ans \right)\right)$, thus matching the above upper bound 
(up to poly-logarithmic factors).
\end{abstract}

\section{Introduction}\label{sec:intro}
Low-degree polynomial approximations of Boolean functions were introduced
by Razborov in his celebrated work~\cite{Razborov1987} on proving
lower bounds for the class of Boolean functions computed by low-depth
circuits. We begin by recalling this notion of approximation over $\reals$.

\begin{definition}[probabilistic degree]\label[definition]{def:pdeg}
Given a Boolean function $f:\zo^n \to \zo$ and $\eps \in (0,1/3)$, an \emph{$\eps$-error
  probabilistic polynomial} over $\reals$\footnote{Similar
  notions over other fields are also studied. Unless otherwise
  specified, we will be considering probabilistic polynomials over
  the reals in this paper.} for $f$ is a
distribution $\distP$ of polynomials $P(x_1,\ldots,x_n)\in \reals[x_1,\ldots,x_n]$ such that
for any $\bar{x}\in \{0,1\}^n$, we have $\Pr_{P\sim \distP}[P(\bar{x}) \neq
  f(\bar{x})] \leq \eps$. The \emph{$\eps$-error probabilistic degree of
    $f$}, denoted by $\pdeg_\eps(f)$, is the smallest non-negative
  integer $d$ such that the following holds: there exists an $\eps$-error
  probabilistic polynomial $\distP$ over $\reals$ such that $\distP$ is
  entirely supported on polynomials of degree at most $d$. 
\end{definition}

Classical results in polynomial approximation of Boolean
functions~\cite{TodaO1992,Tarui1993,BeigelRS1991} show that the $\OR$
function over $n$ variables, denoted by $\OR_n$, has $\eps$-error
probabilistic degree at most $\bigo{\log n \cdot \lepsinv}$. This
basic construction for the $\OR$ function is then recursively used to
show that any function computed by an $\AC^0$ circuit of size $s$ and
depth $d$ has $\eps$-error probabilistic degree at most
$(\log s)^{\bigo{d}}\cdot \lepsinv$ (see work by Harsha and
Srinivasan~\cite{HarshaS2019-ac0} for recent improvements). These
results can then be used to prove~\cite{Razborov1987,Smolensky1987} a
(slightly weaker) version of H{\aa}stad's celebrated
theorem~\cite{Hastad1989} that parity does not have
subexponential-sized $\AC^0$ circuits. These results were employed
more recently by Braverman~\cite{Braverman2010} to prove that
polylog-wise independence fools $\AC^0$ functions.

Despite the fact that probabilistic polynomials for the $\OR$ function
are such a basic
primitive, it is surprising that we do not yet have a complete
understanding of $\pdeg_\eps(\OR_n)$. As mentioned above, it is known from
the works of Beigel, Reingold and Spielman~\cite{BeigelRS1991} and
Tarui~\cite{Tarui1993} that $\pdeg_\eps(\OR_n) = \bigo{\log n \cdot \lepsinv}$.  
The Schwartz-Zippel lemma implies that a dependence of $\bigomega{\lepsinv}$ is necessary in the above bound.
However, until recently, it wasn't clear whether any dependence on $n$ is necessary
in $\pdeg_\eps(\OR_n)$ over the reals\footnote{For finite fields of constant size,
Razborov~\cite{Razborov1987} showed that the $\varepsilon$-error
probabilistic degree of $\OR_n$ is $\bigo{\lepsinv}$,
independent of $n$, the number of the input bits.}. In recent papers of
Meka, Nguyen and Vu~\cite{MekaNV2016} and Harsha and Srinivasan~\cite{HarshaS2019-ac0}, 
it was shown using anti-concentration of low-degree polynomials that the 
$\pdeg_{\quarter}(\OR_n) = \widetilde{\Omega}(\sqrt{\log n})$. The main objective of this paper
is to obtain a better understanding of the $\eps$-error probabilistic
degree of $\OR_n$, $\pdeg_\eps(\OR_n)$. In addition to being
interesting in its own right, this question has bearing on the amount
of independence needed to fool $\AC^0$ circuits. Recent improvements
due to Tal~\cite{Tal2017} and Harsha and Srinivasan~\cite{HarshaS2019-ac0} of Braverman's result demonstrate that
$(\log s)^{2.5d +\bigo{1}}\cdot \lepsinv$-wise independence fools functions
computed by $\AC^0$ circuits of size $s$ and depth $d$ up to error $\varepsilon$. An improvement
of the upper bound on $\pdeg_\eps(\OR_n)$ to $\bigo{\log n +
\lepsinv}$ could potentially strengthen this result to $(\log s)^{d +\bigo{1}}\cdot
\lepsinv$, nearly matching the lower bound of $(\log s)^{d-1}\cdot \lepsinv$ due to Mansour~\cite{LubyV1996}. 

The above discussion demonstrates that the current bounds on $\pdeg_{\varepsilon}(\OR_n)$ 
fall short of being tight in two aspects: one, the dependence on $n$ in the lower bound is
$\Omegatilde{\sqrt{\log n}}$ while in the upper bound it is $\bigo{\log
n}$ and two, the joint dependence on $\eps$ and $n$ in the upper bound
is multiplicative, i.e., $\bigo{\log n \cdot \lepsinv}$ while the
current lower bounds can only show an additive
$\Omegatilde{\sqrt{\log n} + \lepsinv}$ bound.

Which of these bounds is tight? A casual observer might suspect that the upper bound is, 
given the relatively neat expression. However, a closer look tells us  that it cannot be, 
at least when $\varepsilon$ is quite small. For example, setting $\varepsilon = 1/2^{\Omega(n)}$, 
the upper bound yields a degree of $O(n\log n),$ but 
it is a standard fact 
that any Boolean function 
on $n$ variables can be represented exactly (i.e., with no error) as a polynomial of degree $n$. 
Hence the upper bound is not tight in this regime. 

Our first observation is that the upper bound of Tarui and Beigel et al.~\cite{BeigelRS1991} can indeed 
be slightly improved to $\bigo{\ans}$\footnote{Here, $\binom{N}{\leq \alpha}$
denotes  $\sum_{0 \leq i \leq \alpha} \binom{N}{i}$. We use the
convention that $\binom{N}{i} = 0$ if $i > N$.}; note that this is asymptotically better than  
$\bigo{\log n \cdot \lepsinv}$ for very small $\eps.$ This interpolates smoothly between the construction of 
Tarui~\cite{Tarui1993} and Beigel et al.~\cite{BeigelRS1991} and the exact representation of degree $n$ mentioned above. 
(See \cref{sec:upper} for details on this upper-bound construction.) 

Given this observation, one might hope
to prove a matching lower bound on the $\eps$-error
probabilistic degree of $\OR_n$. We can indeed show such a bound (up to polylogarithmic factors) if we suitably
restrict the class of polynomials being considered. While restricted,
this subclass of polynomials nevertheless includes all polynomials
that were used in previous upper bound constructions, including our
own. Moreover, this result generalizes a result of Alon, Bar-Noy,
Linial and Peleg~\cite{AlonBLP1991}, who prove such a result for a
further restricted class of polynomials (mentioned at the end of this
section) and for $\lepsinv = O(\log n)$\footnote{The result
  of~\cite{AlonBLP1991} is stated in a slightly different language,
  but is essentially equivalent to a probabilistic degree lower bound
  for $\OR_n$ for a suitable class of polynomials.}.  A careful
reworking of their analysis shows that their lower bound extends to
even smaller $\epsilon$ to show a lower bound of $\Omega(\ans)$ for this smaller class of polynomials. 

To state our result, we first need to describe the class of
polynomials for which our bounds hold. To this end, we note that all
known upper-bound constructions of probabilistic polynomials for the
$\OR$ function have the following structure:
\[
P(x_1,\dots,x_n) = 1 - \prod_{i \in [t]} \left(1- L_i(x_1,\dots,x_n)\right),
\]
where each $L_i(x_1,\dots, x_n) = a_{i1} x_1 + a_{i2} x_2 + \cdots +
a_{in}x_n$ is a linear form in the variables $x_1,\ldots,x_n$ (here, $a_{ij}
\in \reals$). 
This motivates
the following definition.

\begin{definition}[hyperplane covering polynomials]
\label{def:hcpdeg}
A polynomial $P \in \reals[x_1,\ldots,x_n]$ is said to be a hyperplane
covering polynomial of degree $t$ if there exist $t$ linear forms $L_1, \ldots, L_t$
over the reals such that 
\[
P(x_1,\ldots, x_n ) = 1 - \prod_{i \in
  [t]}\cbra{1-L_i(x_1,\ldots,x_n)} .
\]
For $\eps \in (0,1/2)$, the \emph{$\eps$-error hyperplane covering probabilistic degree of
    $f$}, denoted by $\hcpdeg_\eps(f)$, is the smallest non-negative
  integer $d$ such that the following holds: there exists an $\eps$-error
  probabilistic polynomial $\distP$ over $\reals$ such that $\distP$ is
  supported on hyperplane covering polynomials of degree at most $d$. 
\end{definition}
We call these polynomials hyperplane covering polynomials as these
polynomials have the property that the set of points in the Boolean hypercube where 
the polynomial evaluates to 1 (i.e, the set $\{ \bar{x} \in \zo^n \mid P(\bar{x}) = 1\}$) 
is a union of hyperplanes not passing through the origin. 
We further note that all these
polynomials satisfy the property that $P(\bar{0}) = 0$. Since hyperplane covering 
polynomials are a subclass of probabilistic polynomials,
$\hcpdeg_\eps(f) \geq \pdeg_\eps(f)$. Since all our upper-bound
constructions for the OR polynomials are hyperplane covering
polynomials, we not only have that $\pdeg_\eps(\OR_n) = \bigo{\ans}$ but also that $\hcpdeg_\eps(\OR_n) = \bigo{\ans}$. Our main result is the following (almost) tight bound on the 
$\eps$-error hyperplane covering
probabilistic degree of the $\OR$ function.

\begin{theorem}[hyperplane covering degree of
  $\OR_n$]\label{thm:hpdeg_lbd}
For any any positive integer $n$ and $\eps \in (0,1/3)$,
\[
\hcpdeg_\eps(\OR_n) = \lb .
\]
\end{theorem} 

It is open if this result can be extended to prove a tighter lower bound on
the $\eps$-error probabilistic degree of $\OR_n$. The special class of hyperplane covering polynomials 
for which Alon, Bar-Noy, Peleg and Linial~\cite{AlonBLP1991} proved a similar bound is the class of 
hyperplane covering polynomials where the linear forms are sums of variables 
(i.e., $L_i(\bar{z}) = \sum_{j \in S_i}z_j$ for some $S_i \subseteq [n]$)
Ideally, one would have liked to extend their lower bound result for
hyperplane covering polynomials where the linear forms are sums of
variables to {\em all} polynomials. \cref{thm:hpdeg_lbd}, is a step in
this direction, in that, it shows that
their result can be extended to a slightly larger class, the set of
all hyperplane covering polynomials (modulo polylogarithmic factors). We remark that though our 
lower bound works for a larger class of polynomials, our proof technique is nevertheless inspired by their proof. 

\paragraph*{Organization:} The rest of the paper is organized as
follows. After some preliminaries, we prove our improved upper bound
(\cref{thm:our ub}) in \cref{sec:upper} and prove the lower bound
(\cref{thm:hpdeg_lbd}) in \cref{sec:lower}.

\section{Preliminaries}

\paragraph{Notation:} For a string $x \in \zo^n$, we denote by $|x|$,
the Hamming weight of $x$. The $i$-Hamming slice will refer to the set
of strings $x$ such that $|x| = i$. For a set $S$, $|S|$ denotes the
cardinality of $S$. 

Recall the definition of $\eps$-probabilistic degree $\pdeg_\eps(f)$ from the introduction. The following propositions lists some basic properties of the probabilistic degree. 
\begin{proposition}\label{prop:probdegree}
\begin{enumerate}
    \item\label{item:a} $\pdeg_0(\OR_n) = n$.
    \item\label{item:b} If $ 0 \leq \eps \leq \eps' \leq \nicefrac13$, then $\pdeg_\eps(f) \geq \pdeg_{\eps'}(f)$.
    \item\label{item:c} For all $\eps \in [0,\nicefrac1{2^n})$, $\pdeg_\eps(f) = \pdeg_0(f)$.
    \item\label{item:d} For all constant $k$, $\pdeg_{\eps^k}(f) \leq k\cdot \pdeg_\eps(f)$.
    \item\label{item:e} For all $\eps \in [0,\nicefrac1{2^{0.1n}}]$, $\pdeg_\eps(\OR_n) = \Omega(n)$.
\end{enumerate}
\end{proposition}
\begin{proof}
\cref{item:a,item:b,item:c,item:d} follow from definition. For \cref{item:e}, we note that for $\eps \in [0,\nicefrac1{2^{0.1n}}]$, $\pdeg_\eps(\OR_n) \geq \pdeg_{\nicefrac1{2^{0.1n}}}(\OR_n) \geq 1/20 \cdot \pdeg_{\nicefrac1{2^{2n}}}(\OR_n) = 1/20 \cdot \pdeg_0(\OR_n) = \Omega(n)$. 
\end{proof}

Our notion of hyperplane covering polynomials
depends on the notion of a linear form.
\begin{definition}[linear form and its support]
	\label[definition]{def: support_linear_form}
	A linear form $L(x)$ is a homogenous degree one polynomial $a_{1} x_1 + a_{2} x_2 + \cdots +
a_{n}x_n$. Given a linear form $L(x) =a_{1} x_1 + a_{2} x_2 + \cdots +
a_{n}x_n$ we define the support of $L$, denoted as $\supp(L)$, to be
the set of variables $x_i$ whose corresponding coefficient $a_i$ in
$L$  is non-zero.
	\end{definition}	
	
Recall the notion of hyperplane covering polynomials and hyperplane covering probabilistic
degree $\hcpdeg_\eps(f)$ from the introduction. The following proposition is proved similarly to \cref{prop:probdegree}.
\begin{proposition}\label{prop:hcprobdegree}
\begin{enumerate}
    \item\label{item:A} $\hcpdeg_0(\OR_n) = n$.
    \item\label{item:B} If $ 0 \leq \eps \leq \eps' \leq \nicefrac13$, then $\hcpdeg_\eps(\OR_n) \geq \hcpdeg_{\eps'}(\OR_n)$.
    \item\label{item:C} For all $\eps \in [0,\nicefrac1{2^n})$, $\hcpdeg_\eps(\OR_n) = \hcpdeg_0(\OR_n)$.
    \item\label{item:D} For all constant $k$, $\hcpdeg_{\eps^k}(\OR_n) \leq k\cdot \hcpdeg_\eps(\OR_n)$.
    \item\label{item:E} For all $\eps \in [0,\nicefrac1{2^{0.1n}}]$, $\hcpdeg_\eps(\OR_n) = \Omega(n)$.
\end{enumerate}
\end{proposition}

Hyperplane covering
polynomials have the following closure property.
\begin{claim}
\label{claim: or of hcps is hcp}
	Let $t$ be a positive integer. For each $i \in [t]$, let $P_i$ be a hyperplane covering polynomial of degree $d_i$.  Let $P = 1 - \prod_{i \in [t]} (1 - P_i)$. Then $P$ is a hyperplane covering polynomial of degree at most $\sum_{i \in [t]} d_i$.		 	
\end{claim}
\begin{proof}
	For all $i \in [t]$, since $P_i$ is hyperplane covering polynomial, there exist linear forms $L_{i,1},\dots , L_{i,d_i}$  such that 
	$P_i= 1 - \prod_{j \in [d_i]}(1 - L_{i,j})$. 
	\begin{align*}
		P &= 1 - \prod_{i \in [t]}(1 - P_i)
		= 1 - \prod_{j \in [d_j]}(1 -L_{i,j}).
	\end{align*}
	Therefore by \cref{def:hcpdeg}, $P$ is a hyperplane covering polynomial of degree at most $\sum_{i \in [t]}d_i$.	
\end{proof}

The proof of our lower bound requires the following variant of the
Schwartz-Zippel Lemma (due to Alon and F\"{u}redi~\cite{AlonF1993})
and Littlewood-Offord-Erd\"{o}s' anti-concentration lemma of linear
forms over the reals, which we state below.

    \begin{lemma}[{\cite[Theorem~5]{AlonF1993}}]
	\label[lemma]{lem:sz}
		Let $P\in \reals[x_1,\ldots,x_n]$ be a
                polynomial of degree at most $d$\\
                 over $\reals$ computing a non-zero function over $\{0,1\}^n$. Then for $x$ chosen uniformly from $\{0,1\}^n,$
			\[\prob{x \in \zo^n}{P(x) \neq 0} \geq \frac{1}{2^d}\ .\]
	\end{lemma}

	\begin{lemma}[\cite{LittlewoodO1938,Erdos1945}]
	\label[lemma]{lem:lo}
		Let $L(x_1,\ldots,x_k) = \sum a_i x_i$ be a linear
                form which is supported on exactly $k$ variables
                (i.e., $a_i \neq 0, i = 1,\ldots,k$). Then, for all $a
                \in \reals$ and $x$ chosen uniformly from $\{0,1\}^n,$
	        \[\prob{x\in \{0,1\}^n}{L_i(x) = a} \leq \frac{1}{\sqrt{k}}\ .\]
	\end{lemma}

Our lower and upper bounds will involve expressions of the form $(\log
n - \log \log (\nicefrac 1\eps))\cdot \log (\nicefrac1\eps)$. The
following claim lets us rewrite this expression more compactly in
terms of binomial coefficients. 

\begin{claim}\label{clm:binomial} For $\eps \in [\nicefrac1{2^{n/2}}, 1/2)$,
  we have
  $(\log n - \log
  \log (\nicefrac1\eps))\cdot \log(\nicefrac 1\eps) = \Theta\left(\log
  \binom{n}{\leq \log(\nicefrac1\eps)}\right).$
\end{claim}
\begin{proof}
Consider an integer $k\in [1,\nicefrac{n}{2}]$. Then we have the following well known bounds: 
\[\left(\frac{n}{k}\right)^k \leq \binom{n}{\leq k}\leq \left(\frac{en}{k}\right)^k.\]
The claim now follows directly with $k=\lfloor\log(\nicefrac1 \eps)\rfloor$, which is between $1$ and $\nicefrac{n}{2}$. (Notice that  $\log n -\log\log(\nicefrac{1}{\eps})=\Theta(\log n- k)$.)


  \end{proof}

\section{Upper bounds on probabilistic degree of OR}\label{sec:upper}
	Prior to this work, the best known construction of a probabilistic polynomial of $\OR_n$ in terms of degree was due to Beigel, Reingold and Spielman~\cite{BeigelRS1991} and Tarui~\cite{Tarui1993}.
	\begin{theorem}[\cite{BeigelRS1991,Tarui1993}]\label{thm:BRSTar}
		For any positive integer $n$ and $\eps \in (0,1/3)$,
			$$\pdeg_\eps(\OR_n) \leq \hcpdeg_{\eps}(\OR_n) = O(\log{n}\cdot \log{1/\eps}).$$	
	\end{theorem}
	Note that since every Boolean function can be represented exactly by a polynomial of degree $n$, the above upper bound is meaningful only when $\eps \ge \frac{1}{2^{n/\log n}}$.     
	We modify the construction of Beigel, Reingold and Spielman~\cite{BeigelRS1991} and Tarui~\cite{Tarui1993} and give a strictly better upper bound in terms of probabilistic degree.  
	\begin{theorem}
	\label{thm:our ub}
			For any positive integer $n$ and $\eps \in (0,1/3)$,
			$$\pdeg_\eps(\OR_n) \leq \hcpdeg_{\eps}(\OR_n) = \bigo{\ans}.$$
	\end{theorem}


  
  	To begin with, we observe that the following
	hyperplane covering polynomial of degree $n$ exactly computes $\OR_n$
	everywhere on the Boolean hypercube:
	\begin{equation}
	\label{eq:OR slice polynomial}
		P_{\OR}(x) := 1- \prod_{i =1}^n \cbra{1- \frac1i\sum_{j \in [n]} x_j}.
	\end{equation}
	For each $i \in [n]$,
	the degree $1$ polynomial $\frac1i\sum_{j \in [n]}x_j$ outputs
        $0$ on the zero input and $1$ on the $i$-th Hamming slice. $P_\OR$ outputs 1 if any of these degree 1 polynomials output 1. 
	
	

	We now recall the construction of Beigel, Reingold and
	Spielman~\cite{BeigelRS1991} and Tarui~\cite{Tarui1993}. 
	For each $0 \leq \ell \leq \log n - 1$,	
	they give a hyperplane covering probabilistic polynomial which
        outputs 0 on the zero input $\bar{0}$ and outputs 1 with
        constant probability for all inputs whose Hamming weight is in
        the range $[2^\ell,2^{\ell+1}]$.  
		
	\begin{lemma}[\cite{BeigelRS1991,Tarui1993}]
	\label{lem:epoch_ hcp}
		Let $n$ be a positive integer. For all integers $\ell$ such that $0 \leq \ell \leq \log n - 1$  and for all $\eps \in (0,1/3)$, there exists a distribution $\mathbf{P}_\ell$ on hyperplane covering polynomials of degree $O(\log {(1/\eps)})$ such that 
	\begin{itemize}
		\item $P(0^n)=0$ for all $P \sim \mathbf{P}_\ell$.
		\item for all inputs $x\ in\ \zo^n$ whose Hamming
                  weight is in the range $[2^\ell,2^{\ell+1}]$ , 	
			$$\prob{P \sim \mathbf{P}_\ell}{P(x) = 1} \geq 1-\eps.$$
	\end{itemize}							 
	\end{lemma}
\begin{proof} Fix $0 \leq \ell \leq \log n$  and $\eps \in
  (0,1/3)$. 
  
	We begin by defining a distribution $\distL_\ell$ of linear
        forms as follows: pick a random set $S \subseteq [n]$ by picking each element of $[n]$ independently with probability $\frac{1}{2^{\ell}}$ and construct the linear polynomial 
    \[L_S(x) := \sum_{i \in S} x_i\enspace.\]
    For a non-zero input $x = (x_1,\ldots,x_n)$ such that the Hamming
    weight $|x|$ of $x$ is in $[2^\ell, 2^{\ell+1}]$, we have
        \begin{align*}
            \prob{S}{L_S(x) = 1} &= |x|\cbra{\frac{1}{2^\ell}} \cbra{1-\frac{1}{2^\ell}}^{|x|-1}&&[\text{where\ }0^0=1]\\
                    &= \frac{|x|}{2^\ell} \exp{(-O(1))}&&[\because (1-a)^b\geq \exp{(-\nicefrac{ab}{1-a})}]\\
                    &= \Omega(1)\enspace.
        \end{align*}
	
      In order to get a probabilistic polynomial $\mathbf{P}_\ell$
      which satisfies the requirements of \cref{lem:epoch_ hcp}, we sample $t =O(\log {(1/\eps)})$ linear forms $L_1, \dots, L_t$ 
      independently from $\distL_\ell$ and construct the polynomial
	$$P_{L_1,\dots ,L_t}(x) := 1 - \prod_{i \in [t]}(1 - L_i(x))$$

Note that all the polynomials in the support of $\mathbf{P}_\ell$ are
hyperplane covering polynomials. For any $L_1,\dots ,L_t$, degree of
$P_{L_1,\ldots ,L_t}$ is $t = O(\log{1/\eps})$.  
	$P_{L_1,\ldots,L_t}(0^n) = 0$ since $L_i(0^n) = 0$ for all $i \in [t]$. 
	For any input $x$ such that $|x| \in [2^\ell,2^{\ell+1}]$,
        $P_{L_1,\ldots,L_t}$ errs on $x$ only if for all $i \in [t]$,
        $L_i(x) \neq 1$, which happens with probability at most
        inverse exponential in $t$ and hence at most $\eps$ (since
        $\prob{L_i \sim \distL_\ell}{L_i(x) \neq 1}$ is at most some
        constant less than 1 for each $i$).
\end{proof}

	
	\cref{thm:BRSTar} is obtained by considering the following probabilistic polynomial $\mathbf{P}$. 
	For each $\ell \in \{0, \dots ,\log n -1\}$, sample $P_\ell \sim \mathbf{P}_\ell$ and construct
		$$P := 1 -\prod_{\ell \in [\log n]} (1 - P_\ell).$$
This construction uses the probabilistic polynomial of degree $\lepsinv$
from \cref{lem:epoch_ hcp} for each of the $\log n$ epochs (where the
$\ell$-th epoch refers $\{ x \mid |x| \in [2^\ell,2^{\ell+1}\}$). This
turns out to be wasteful for the lower epochs ($\ell \leq
\llepsinv$). We observe that since the lower epochs have fewer slices,
we can gain by using the polynomial construction from \eqref{eq:OR
  slice polynomial} instead.     	
  
	\begin{proof}[Proof of \cref{thm:our ub}]
		Consider the following distribution $\mathbf{P}$ on
                hyperplane covering polynomials $P$: 
For each $\ell \in [\llepsinv,\log n - 1]$, sample $P_\ell \sim
\mathbf{P}_\ell$ independently and construct the polynomials $P', P''$
and $P$ as follows.
	\begin{align*}
		P'(x) &:= 1 -\prod_{\ell = \llepsinv}^{\log n -1}
						\cbra{
							1-P_\ell(x)
						},\\	
		P''(x) &:= 1 -\prod_{i =1}^{\lepsinv} 
					\cbra{
						1- \frac1i \sum_{j \in [n]} x_j
					},\\
		P(x) &:= 1 - (1-P'(x))(1- P''(x)) .
	\end{align*}		
	
	\begin{itemize}
		\item By \cref{lem:epoch_ hcp}, $P_\ell$ is a hyperplane covering polynomial for each $\ell \in [\llepsinv , \log n - 1]$. Therefore by \cref{claim: or of hcps is hcp},  $P'$  is hyperplane covering polynomial. Since $P''$ is also a hyperplane covering polynomial, so is $P$.    
		\item Observe that $P''(0^n) = 0$. 
	By \cref{lem:epoch_ hcp}, for all $\ell \in [\llepsinv , \log n - 1]$,
	$P_\ell(0^n) = 0$ and hence $P'(0^n) = 0$.
		\item For all $i \leq \lepsinv$, for all inputs $x$ 	
	from the $i$-th slice, $P''(x) = 1$ and hence $P(x) = 1$.
	For each $\ell \in [\llepsinv , \log n -1]$ and each input $x$ from the $\ell$-th epoch, 
		$$\prob{P_\ell \sim \mathbf{P}_\ell}{P_\ell(x) =1} \geq 1 -\eps.$$ 
		Since $P_\ell(x) = 1$ implies $P(x) =1$, 
			$$\prob{P \sim \mathbf{P}}{P(x) =1} \geq \prob{P_\ell \sim \mathbf{P}_\ell}{P_\ell(x) =1} \geq 1- \eps.$$	
	Therefore for all nonzero inputs $x$, $\prob{}{\mathbf{P}(x) = 1} \geq 1 - \eps $.	
	\item Since $P''$ has degree $\lepsinv$ and by
          \cref{lem:epoch_ hcp}  $P_\ell$ has degree at most
          $\lepsinv$ for each $\ell \in [\llepsinv, \log n -1]$, $P$ has
          degree $(\log n - \llepsinv ) \lepsinv = \bigo{\ans}$
          (using \cref{clm:binomial}). 
	\end{itemize}
	Therefore $\mathbf{P}$ is an $\eps$-error probabilistic polynomial for $\OR_n$ supported of hyperplane covering polynomials of degree $\bigo{\ans}$.		
\end{proof}

\section{Lower bound on hyperplane covering degree of OR}\label{sec:lower}
    We now turn to the lower bound. 
    To prove a lower bound of $d_\eps := \Omegatilde{\ans}$,  
    by Yao's minimax theorem (duality arguments) it suffices (and is necessary) to
    demonstrate a ``hard'' distribution $\calD_{\eps}$ on $\zo^n$ under which it is
    hard to approximate $\OR_n$ by any hyperplane covering
    polynomial of degree at most $d_\eps$. 

    Similar to previous
    works~\cite{MekaNV2016,HarshaS2019-ac0}, our choice of hard distribution is
    motivated by the polynomial constructions in the upper bound. 
    Our hard distribution $\calD_\eps$ is defined in terms of the following two distributions.
\begin{definition}[$p$-random assignment] 
    \label[definition]{def:zo-restriction}
		Let $p \in [0,1]$ and let $X = \{x_1, \ldots ,x_n\}$
                be 		a set of variables. A $p$-random
                assignment of $X$, denoted by $\mu_p^X$, is a random
                assignment $\mu\colon X \to \zo$ that is chosen as
                follows:  for each of the variables $x_i \in X$
                independently set $\mu(x_i)$ to 1 with probability $p$ and 0 with probability $1-p$.  	    
\end{definition}
    \begin{definition}[$p$-random $\zs$-restriction]
    \label[definition]{def:zs-restriction} 
	Let $p \in [0,1]$ and let $X= \{x_1, \ldots ,x_n\}$ be a set
        of variables. A p-random $\zs$-restriction of $X$, denoted by
        $\rho_p^X$, is a random restriction $\rho\colon X \to \{0,*\}$
        chosen as follows:  for each of the variables $x_i \in X$
        independently set $\rho(x_i)$ to 0 with probability $(1-p)$ and
        $*$ with probability $p$ (i.e., the variable is unset with
        probability $p$).
    \end{definition}
When the set of variables $X$ is clear from the context, we will drop the
superscript $X$ and refer to the corresponding distributions as
$\mu_p$ and $\rho_p$ respectively. 
	We observe that $\mu^X_p$ can be generated by sampling a
        $2p$-random $\zs$-restriction $\rho^X_{2p}$ and setting the
        unset variables (i.e., $\left(\rho^X_{2p}\right)^{-1}(*)$)
        according to $\mu_{1/2}$. In short, 
		$$\mu^X_{p} = \mu_{1/2}^{\left(\rho^X_{2p}\right)^{-1}(*)} \circ
                \rho^X_{2p}.$$ 		
	We use this observation crucially later on in the proof of the lower bound.        

    \begin{definition}[hard distribution]
	\label[definition]{def:hard_distribution}
		For $\eps \in [\nicefrac1{2^{n/2}},\nicefrac12)$, consider the distribution $\calD_\eps$ on the input set $\zo^n$ defined as follows:
    	 Let $I_{\eps} := [1,\log n - \llepsinv]\cap \mathbb{Z}$. 
		Pick an $\ell$ uniformly at random from $I_{\epsilon}$
                and output a random sample $x$ from $\mu_{1/2^\ell}$, i.e., for each $i \in [n]$,
                independently set $x_i\gets 1$ with probability $\nicefrac1{2^\ell}$ and $0$ otherwise.  
	\end{definition}
    
    The hard distribution $\calD_\eps$ is a convex combination of the
    distributions $\mu_{\nicefrac1{2^{\ell}}}$ for $\ell \in I_\eps$. In other words, $\calD_\eps := \frac{1}{|I_\eps|}
    \sum_{\ell \in I_\eps} \mu_{\nicefrac1{2^{\ell}}}$. Each of
    the distributions $\mu_{\nicefrac1{2^{\ell}}}$ roughly correspond to the
    epochs used in the upper-bound construction.  The following claim
    shows that the distribution
    $\calD_\eps$ puts probability at most $\eps$ on the all-zeros
    input $\bar{0}$.
\begin{claim}\label{clm:zero} For $\eps \geq  \nicefrac1{2^{n/2}}$, we have $\calD_\eps(0^n)
  \leq \eps$.
\end{claim}
\begin{proof}
$\calD_\eps$ is generated by drawing an $\ell$ from $I_\eps$ at random and returning a draw from $\mu_{1/2^\ell}$.
 Since 
 $$\mu_{1/2^\ell}(0^n) = (1-1/2^\ell)^n \leq \eps$$
  for $\ell \leq \log n - \llepsinv$, $\calD_\eps(0^n) \leq \eps$.
\end{proof}



   \cref{thm:hpdeg_lbd} follows from the following ``distributional''
   version of the theorem for $\eps \in [\nicefrac1{2^{n/2}},\nicefrac13]$. For smaller $\eps$, \cref{thm:hpdeg_lbd} follows from \cref{prop:hcprobdegree}:\cref{item:E}.
   
   \begin{theorem}
	\label{thm: deterministic_polynomial_lb_wrt_hard_dist}
        Let $\eps \in [\nicefrac1{2^{n/2}},\nicefrac13]$ and $\calD_\eps$ be the hard distribution defined in
        \cref{def:hard_distribution} and $P = 1- \prod_{i \in [t]}\cbra{1-L_i}$ be a hyperplane covering polynomial of degree $t$ such that       
            \[\prob{x \sim \calD_\eps}{P(x) \neq OR_n(x)} \leq \eps\]
        then, $t \geq \lb$. 
      \end{theorem}

         The rest of this section is devoted to proving \cref{thm: deterministic_polynomial_lb_wrt_hard_dist}. We begin with a
        proof outline in \cref{sec:lboutline} followed by the proof in \cref{sec:lbproof}.

\subsection{Proof outline}\label{sec:lboutline}  
    We would like to show that hyperplane covering polynomial $P$ that
    approximates $\OR_n$ with respect to the distribution $\calD_\eps$ (as in
    \cref{thm: deterministic_polynomial_lb_wrt_hard_dist}) must have large degree. Let $\calL$ denote the set of linear
    forms that appear in $P$, i.e., $\calL := \{ L_i \mid i \in [t]\}$. 
    
    Let us see how $P$ behaves on the
    distribution $\mu_{\nicefrac1{2^{\ell}}}$. 
    or equivalently $\mu_{\half} \circ
    \rho_{\nicefrac1{2^{\ell-1}}}$.  
    Let us see what happens to the linear forms in $\calL$
    when the restriction $\rho\sim \rho_{\nicefrac1{2^{\ell-1}}}$ is first applied. We first consider two extreme cases.
    
    \begin{description}
        \item[Very few linear forms survive:] Suppose all but $\log
            (\nicefrac1{2\eps})$ linear forms trivialize on the restriction $\rho$ (i.e., the
            corresponding linear form $L_i|_{\rho}$ becomes 0). Then,
            $(1-P)|_{\rho}$ is a polynomial of degree at most $\log
            (\nicefrac1{2\eps})$ computing a non-zero function (since $1-P(\bar{0}) = 1$). Hence, by \cref{lem:sz}, it is not equal to 0
            with probability at least $2\eps$. 
            This implies that the polynomial
            $P$ errs with probability at least $\eps$ on the distribution $\mu_{\nicefrac1{2^{\ell}}}$.
        \item[All linear forms that survive have large support:] Suppose all the
            linear forms that survive post restriction $\rho$ have large
            support, say $4t^2$. Then, by the anti-concentration of linear
            forms over reals (\cref{lem:lo}), we have that
            each linear form is 1 with probability at most
            $\nicefrac1{\sqrt{4t^2}} = \nicefrac1{2t}$. Since there
            are at most $t$ linear forms, the
            probability that any of them is 1 is at most
            $\nicefrac{t}{2t}  = \half$. Thus, $P$ errs with probability  $\half$ on the distribution 
            $\mu_{\nicefrac1{2^{\ell}}}$.
    \end{description}

    Note that the actual situation for each distribution
    $\mu_{\nicefrac1{2^{\ell}}}$ will most likely be a combination of the above
    two. We can then show that a combination of the above two arguments
    will still work if the surviving linear forms have the following nice
    structure. Let $\calL_{\rho}$ be the set of surviving linear forms
    subsequent to the restriction $\rho$, i.e., $\calL_\rho : = \{ L_i|_\rho
    \mid i \in[t], L_i|_\rho \neq 0\}$. Suppose $\calL_\rho$ can be
    partitioned into 2 sets $\calL_\rho' \disjunion \calL_\rho''$ such that
    the number of linear forms in $\calL_\rho'$ is small (less than
    $O(\lepsinv)$) and each of the linear forms in $\calL_\rho''$
    have large support even after removing $\union_{L \in
    \calL_\rho'}\supp(L)$ from their support. How does one then show that a constant
    fraction of $\rho$'s satisfy that the corresponding linear forms
    $\calL_\rho$ have this nice structure? For this, we draw inspiration
    from the proof of Alon, Bar-Noy, Linial and
    Peleg~\cite{AlonBLP1991}, where they prove similar bounds for
    hyperplane covering polynomials supported entirely on linear forms
    arising as sums of variables. They construct an appropriate
    potential function that guarantees a similar property in their
    lower-bound argument. 
    
    We use a slightly different potential function $\Phi_\ell(\calL)$, which has the following nice property. If the total number of
    linear forms is $t$, then
    $\E{\ell}{\Phi_\ell(\calL)} = O(t/(\log n - \llepsinv))$ and
    furthermore, whenever $\Phi_\ell(\calL)$ is small then the
    corresponding set $\calL_\ell$ of surviving linear forms post
    restriction $\rho_{\nicefrac1{2^{\ell-1}}}$ can be partitioned as
    indicated above. This shows that for most $\ell$, $P$ errs on
    computing  $\OR_n$ unless $t$ is large. 

\subsection{Proof of \cref{thm: deterministic_polynomial_lb_wrt_hard_dist}}\label{sec:lbproof} 

We now turn to defining the potential function $\Phi_\ell(\calL)$, indicated in
the proof outline.
    
    \begin{definition}[potential function]
    \label[definition]{def: potential function}
The weight of a linear form $L$, denoted by $w(L)$, is defined as follows:
    	    \[w(L) :=  
    	                    \begin{cases}
    	                        0 &\text{if } \supp(L) = \emptyset,\\
    	                        \frac{1}{\log^2{(2|\supp(L)|)}}& \text{otherwise}.
    	                    \end{cases}\]

	    Given a collection $\calL = \{L_1, \ldots , L_t\}$ of linear
            forms and $\ell$ a positive integer, the potential
            function $\Phi_\ell(\calL)$ is defined as follows
    	\[\Phi_\ell(\calL) :=
          \sum_{i =1}^t \E{\rho\sim \rho_{\nicefrac1{2^{\ell-1}}}}{w\cbra{L_i|_{\rho}}}
          , \]
where $\rho_{\nicefrac1{2^{\ell-1}}}$ is a $\zs$-restriction as
defined in \cref{def:zs-restriction}.
	\end{definition}

The potential function $\Phi_\ell(\calL)$ satisfies the following two
properties, given by \cref{prop:pot_exp,prop:partition}

\begin{proposition}
    \label[proposition]{prop:pot_exp}  
    There exists a universal constant $C$ such that the following holds. 
	    Let $\calL=\bra{L_1,\ldots, L_t}$ be any collection of $t$ 
            linear forms, then
    	    \[\E{\ell \in I_\eps}{\Phi_\ell(\calL)} \leq 
              \frac{Ct}{|I_\eps|}\ .\]   
	\end{proposition}

	\begin{proposition}[partition of linear forms]
	\label[proposition]{prop:partition}
		Let $\calL= \{L_1,\ldots ,L_t \}$ be a collection of
                $t$ non-zero linear forms and $K,R$
                be two positive integers such that \[\sum_{i=1}^t
                w(L_i) < \frac{R}{\log^2 (2RK)}\ .\]Then, there
                exists a partition $\calL = \calL' \disjunion
                \calL''$ of the set of linear forms $\calL$ such that 
        \begin{itemize}
            \item $|\calL'| \leq R$, 
            \item  For all $L \in \calL''$, $|\supp(L) \setminus
              \union_{L' \in \calL'} \supp(L')| \geq K$.
        \end{itemize}
\end{proposition}

Before proving these two propositions, we first show how they imply
\cref{thm: deterministic_polynomial_lb_wrt_hard_dist}. 

\begin{proof}[Proof of {\cref{thm:
      deterministic_polynomial_lb_wrt_hard_dist}}]
Let 
		\[t = \lbdetail,\]
	where $C$ is the universal constant in \cref{prop:pot_exp}.	
	Note that $t = \lbinline$ (via \cref{clm:binomial}).	
	Let $P = 1 - \prod_{i \in [t]}(1-L_i)$ be any hyperplane covering
polynomial of degree $t$. Recall that $\eps \geq
\nicefrac1{2^{n/2}}$. Recall $\calL=\{ L_1,\ldots, L_t\}$. To prove
\cref{thm: deterministic_polynomial_lb_wrt_hard_dist} it suffices to
show the following	
	\begin{equation}
	\label{eq:goal 1}
		\prob{x \in \calD_\eps}{P(x) \neq \OR_n(x)} > \eps.
	\end{equation}	 
	 
We have from \cref{clm:zero} that $\calD_\eps(0^n) = \Pr_{x\sim
  \calD_\eps}[x = \bar{0}] < \eps$. Since $\prob{}{P(x) \neq 1} \leq
\prob{}{P(x) \neq \OR_n(x)} + \prob{}{x \neq 0^n}$, in order to show
inequality \eqref{eq:goal 1}, it suffices to  prove
 	\begin{equation}
 	\label{eq:goal2}
 		\prob{x \sim \calD_\eps}{P(x) \neq 1} \geq 2 \eps.
 	\end{equation}
\noindent
Since
$\calD_\eps = \frac{1}{|I_\eps|}
    \sum_{\ell \in I_\eps} \mu_{\nicefrac1{2^{\ell}}}$ and
    $\mu_p^{[n]} = \mu_{\half}\circ \rho_{2p}^{[n]}$, this is
    equivalent to showing
\begin{equation}\label{eq:contrad}
\E{\ell \in I_\eps}{\E{\rho \sim \rho_{\nicefrac1{2^{\ell-1}}}}{
    \Pr_{x\sim\mu_{\half}}\left[P|_{\rho}\left(x\right)\neq1
 \right]}} \geq
    2\eps\ .
\end{equation}

To this end, we first apply \cref{prop:pot_exp} to the set $\calL$ of
$t$ linear forms in the polynomial $P$ to obtain that
\[\E{\ell \in I_\eps}{\E{\rho\sim
      \rho_{\nicefrac1{2^{\ell-1}}}}{\sum_{i\in [t]}w(L_i|_\rho)}}=\E{\ell \in I_\eps}{\Phi_\ell(\calL)} \leq 
              \frac{Ct}{|I_\eps|}\ .\]   
Applying Markov's inequality to the above, we have
\[ \Pr_{\ell,\  \rho}\left[ \sum_{i \in [t]} w(L_i|_\rho) \leq
      \frac{2Ct}{|I_\eps|}\right] \geq \frac12\ .
\]
We call an$(\ell,\rho)$ pair \emph{good} if the above event holds,
i.e., $\sum_{i=1}^t w(L_i|_{\rho}) \leq
\nicefrac{2Ct}{|I_\eps|}$. Thus, 
\begin{equation}\label{eq:good}
\Pr_{\ell, \ \rho}[ (\ell,\rho) \text{ is good }] \geq \half\ .
\end{equation}

Now given a good $(\ell, \rho)$-pair, let $\calL_{\rho}$ be the set of surviving linear forms
    subsequent to the restriction $\rho$, i.e., $\calL_\rho : = \{ L_i|_\rho
    \mid i \in[t], L_i|_\rho \neq 0\}$. We thus have $\sum_{L \in \calL_\rho}
    w(L) \leq \nicefrac{2Ct}{|I_\eps|}$. Let $K :=4t^2$ and $R: =
    \log(\nicefrac1{8\eps})$. It can be checked that for this choice of parameters we have
    $\nicefrac{2Ct}{|I_\eps|} < \nicefrac{R}{\log^2 (2RK)}$.
We can now apply \cref{prop:partition} to obtain a partition
$\calL_\rho = \calL_\rho' \disjunion \calL_\rho''$ such that
\begin{itemize}
\item $|\calL_\rho'| \leq R = \log(\nicefrac1{8\eps})$,
\item for all $L \in \calL_\rho''$, we have $|\supp(L) \setminus
  \union_{L' \in \calL_\rho'} \supp(L') | \geq K=4t^2$.
\end{itemize}
Consider the polynomial $P|_\rho = 1 - \prod_{i \in [t]} (1-
L_i|_\rho) = 1- \prod_{L \in \calL|_\rho}(1-L)$ subsequent to the restriction $\rho$. We will rewrite this
polynomial as $P|_\rho = 1 - Q'_\rho \cdot Q''_\rho$ where the
polynomials $Q_\rho'$ and $Q_\rho''$ are defined as follows (using the
sets $\calL_\rho'$ and $\calL_\rho''$ respectively). 
\begin{align*}
Q_\rho'(x) &:= \prod_{L \in \calL_\rho'} (1-L(x)),\\
Q_\rho''(x) &:= \prod_{L \in \calL_\rho''} (1-L(x)).
\end{align*}
Note that $P|_\rho = 1- Q_\rho'\cdot Q_\rho''$. 

Since $|\calL_\rho'| \leq \log(\nicefrac1{8\eps})$, we have that the
degree of $Q_\rho'$ is at most $\log(\nicefrac1{8\eps})$. Furthermore $Q'_\rho(x) \not\equiv 0$ (since $Q'_\rho(\bar{0})=1$). Thus applying
\cref{lem:sz}, we have
\[
\Pr_{x \sim \mu_{\half}}\left[ Q_\rho'(x) \neq 0 \right] \geq 8\eps.
\]
Consider any setting of variables in $\union_{L \in
  \calL_\rho'}\supp(L)$ such that $Q_\rho'(x) \neq
0$. Even conditioned on setting all these variables, we know that each $L \in
\calL_\rho''$ still has surviving support of size at least
$4t^2$. Thus, by \cref{lem:lo}, we have for each $L \in \calL_\rho''$,
\[
\Pr_{x \sim \mu_{\half}}\left[ L(x) = 1 \mid Q_\rho'(x)
  \neq 0\right] \leq \frac{1}{\sqrt{4t^2}} = \frac1{2t}.
\]
By a union bound, we have
\[\Pr_{x \sim \mu_{\half}} \left[ Q''_\rho(x) = 0 \mid Q'_\rho(x) \neq
    0 \right] = \Pr_{x \sim \mu_{\half}} \left[ \exists L \in \calL_\rho'',
      L(x) = 1 \mid Q'_\rho(x) \neq
    0 \right] \leq \frac{t}{2t} = \frac12.
\]
Hence, 
\[\Pr_{x \sim \mu_{\half}} \left[ P|_\rho(x) \neq 1 \right] =
  \Pr\left[Q'_\rho(x) \neq 0 \right] \cdot \Pr \left[ Q''_\rho(x) = 0 \mid Q'_\rho(x) \neq
    0 \right] \geq  8\eps \cdot \frac12 = 4\eps.
\]
Finally averaging over all $(\ell, \rho)$ we have from above and \eqref{eq:good} 
\[
\Pr_{x \sim \calD_\eps} \left[ P(x) \neq 1\right] \geq \Pr_{\ell,
  \rho}\left[ (\ell,\rho) \text{ is good } \right] \cdot \Pr\left[
  P|_\rho(x) \neq 1 \mid  (\ell,\rho) \text{ is good } \right ] \geq
\frac12 \cdot 4\eps = 2\eps.
\]
This proves \eqref{eq:goal2} and thus completes the proof of \cref{thm: deterministic_polynomial_lb_wrt_hard_dist}. 
\end{proof}

We are now left with the proofs of
\cref{prop:pot_exp,prop:partition}. We begin with the
proof of \cref{prop:partition}.

	\begin{proof}[Proof of {\cref{prop:partition}}]
	    Consider the following algorithm to obtain the partition
            $\calL = \calL'\disjunion \calL''$. 
	    \begin{enumerate}
		    \item Initialize $\calL' \gets \emptyset $ and
                      $\calL'' \gets \calL$.
		    \item While there exists an $L \in \calL''$ such
                      that $|\supp(L) \setminus \union_{L' \in \calL'}\supp(L')| \leq
                      K$, 
                      \begin{itemize}
                        \item Move such an $L$ from $\calL''$ to
                          $\calL'$ (i.e., $\calL' \gets \calL' \union
                          \{L\}$ and $\calL'' \gets \calL'' \setminus
                          \{L\}$).
                        \end{itemize}
                      \end{enumerate}
Let $\supp(\calL')$ be the union of supports of all linear forms in $\calL'$ $(i.e., \supp(\calL') = \cup_{L \in \calL'} \supp(L))$. 
When the algorithm terminates, we have $|\supp(L) \setminus
\supp(\calL')| \geq K$ for all $L \in \calL''$. 

We now argue that $|\calL'| \leq R$. Each iteration of the while loop adds a linear form $L$ to $\calL'$
with at most $K$ new variables. If the while loop is performed for $T$ iterations, then the support of each $L$ added to $\calL'$ is
at most $TK$. We now argue that $T < R$. If not, then after exactly $R$ iterations of the while loop, we have that
\[
\sum_{L \in \calL} w(L) \geq \sum_{L \in \calL'} w(L) \geq
\frac{R}{\log^2(2RK)},
\]
contradicting the hypothesis of the proposition. Hence $T < R$. The
size of $\calL'$ is the number of iterations of the while loop and is
thus bounded above by $R$. This completes the proof of the
proposition.
\end{proof}

\begin{proof}[Proof of {\cref{prop:pot_exp}}]
    	\begin{align*}
        	\E{\ell \in I_\eps}{\Phi_\ell(\calL)} 
        	            &=\E{\ell \in I_\eps}{\E{\rho \sim \rho_{\nicefrac1{2^{\ell-1}}}}{\sum_{i\in[t]}w(L_i|_{\rho})}}\\
                        &= \frac1{|I_\eps|} \sum_{i\in[t]}
                          \sum_{\ell\in I_\eps} \E{\rho}{w(L_i|_{\rho})}\\
                        &= \frac{1}{|I_\eps|}\sum_{i\in [t]} \bigg  (\underbrace{\sum_{\ell >  \log{|\supp(L_i)|}} \E{\rho}{w(L_i|_{\rho})}}_{T_1}
                        + \underbrace{\sum_{\ell \leq \log{|\supp(L_i)|}} \E{\rho}{w(L_i|_{\rho})}}_{T_2}\bigg).
	    \end{align*}
    $T_1$ and $T_2$ are bounded using \cref{claim:logbinomial1} and
    \cref{claim:logbinomial2} respectively. Hence,
        \begin{align*}
    	    \E{\ell \in I_\eps}{\Phi_\ell(\calL)} &\leq \frac{1}{|I_\eps|} \sum_{i = 1}^{t} \left(2+ \frac{\pi^2}{6} + \frac{e}{e-1}\right)
    	            \leq \frac{t}{|I_\eps|}\cdot \cbra{4 + \frac{\pi^2}{6}}.\qedhere
        \end{align*}
    \end{proof}
    
    \begin{claim}
    \label[claim]{claim:logbinomial1}
        Let $L$ be a linear form such that $|\supp(L)| = k$. Then
            \[\sum_{\ell > \log{k}} \E{\rho \sim \rho_{\nicefrac1{2^{\ell-1}}}}{w(L|_{\rho})} \leq 2.\]
    \end{claim}
    \begin{proof}\abovedisplayskip=-15pt 
    \begin{align*}
			\sum_{\ell > \log{k}} \E{\rho \sim \rho_{\nicefrac1{2^{\ell-1}}}}{w(L|_{\rho})}
          	&\leq \sum_{\ell >  \log{k}}
            			\cbra{\prob{\rho}{|\supp(L|_{\rho})| = 0} \cdot 0 +
            			\prob{\rho}{|\supp(L|_{\rho})| \geq 1} \cdot 1}\\
         	&\leq \sum_{\ell > \log{k}} \cbra{1 - \cbra{1 - \frac{1}{2^{\ell-1}}}^{k}}\\ 
         	&\leq \sum_{\ell > \log{k}}\frac{k}{2^{\ell-1}}
            \qquad\qquad \qquad  [\because (1-x)^n \geq 1-nx, \forall \ 0< x \leq 1 ] 
      \\
			&\leq \sum_{i \geq 0} \frac{1}{2^i}  = 2 .         
			\end{align*}
    \end{proof}
    
    \begin{claim}
    \label[claim]{claim:logbinomial2}
        Let $L$ be a linear form such that $|\supp(L)| = k $. Then
            \[\sum_{\ell = 0}^{\log k}\E{\rho \sim \rho_{\nicefrac1{2^{\ell-1}}}}{w(L|_{\rho})} \leq \frac{\pi^2}{6} + \frac{e}{e-1}.\]
    \end{claim}
    \begin{proof}\abovedisplayskip=-15pt
        \begin{align*}
            \E{\rho}{w(L|_{\rho})} 
            	&\leq \prob{\rho}{|\supp(L|_{\rho})| \geq \frac{k}{2^\ell}} \frac{1}{\log^2{(2k/2^\ell)}} + \prob{\rho}{ |\supp(L|_{\rho})| \leq  \frac{k}{2 \cdot 2^{\ell-1}} }\\
               	&\leq \frac{1}{(\log k - \ell+1)^2} + \exp\cbra{-
                  \frac{k}{4\cdot 2^{\ell-1}}}  \qquad\qquad [\text{By Chernoff bound}]\\
        	\sum_{\ell = 0}^{\log k}\E{\rho \sim \calR_\ell}{w(L|_{\rho})} 
            	&\leq \sum_{\ell = 0}^{\log k} \frac{1}{(\log{k} - \ell+1)^2} +  \sum_{\ell = 0}^{\log k} \exp\cbra{-\frac{k}{2^{\ell+1}}}\\
            	&= \sum_{i=1}^{\log k } \frac{1}{i^2} + \sum_{i = 1}^{\log k} \exp\cbra{-2^{i-1}} \\
                &\leq \frac{\pi^2}{6} +  \frac{e}{e-1}\ .
        \end{align*}
    \end{proof}
\section*{Acknowledgments}
The authors thank Noga Alon for referring them to the paper on radio-broadcast~\cite{AlonBLP1991}.

{\small
\bibliographystyle{prahladhurl}
\bibliography{BHMS-bib}
}

\end{document}